\documentclass[12pt]{article}

\usepackage[dvips]{graphicx}
\usepackage{amsmath,amssymb,amsthm}
\usepackage[colorlinks,citecolor=red,linkcolor=blue,urlcolor=magenta]{hyperref}
\numberwithin{equation}{section}

\newtheorem{lemma}{Lemma}
\newtheorem{theorem}{Theorem}
\newtheorem{prop}{Proposition}

\newcommand{\bbR}{\mathbb{R}}
\newcommand{\mm}{\mathcal{M}}
\newcommand{\ve}{\varepsilon}
\newcommand{\mP}{\mathcal{P}}
\newcommand{\mL}{\mathcal{L}}
\newcommand{\rk}{\mathop {\rm rank}\nolimits}
\newcommand{\sgrad}{\mathop {\rm sgrad}\nolimits}
\newcommand{\Tr}{\mathop {\rm Trace}\nolimits}

\newcommand{\bda}{{\boldsymbol \alpha}}
\newcommand{\bdb}{{\boldsymbol \beta}}
\newcommand{\bdo}{{\boldsymbol \omega}}
\def\rmd{{\rm d}}
\def\bdot{{\boldsymbol \cdot}}
\def\rmi{{\rm i}}

\begin{document}

\title{Extensions of the Appelrot classes\\ for the generalized gyrostat in a double force field}

\author{M. P. Kharlamov}

\date{05.09.2013}

\maketitle



\begin{abstract}
For the integrable system on $e(3,2)$ found by Sokolov and Tsiganov we obtain explicit equations of some invariant 4-dimensional manifolds on which the induced systems are almost everywhere Hamiltonian with two degrees of freedom. These subsystems generalize the famous Appelrot classes of critical motions of the Kowalevski top. For each subsystem we point out a commutative pair of independent integrals, describe the sets of degeneration of the induced symplectic structure. With the help of the obtained invariant relations, for each subsystem we calculate the outer type of its points considered as critical points of the initial system with three degrees of freedom.

MSC 70E05,70E17,37J15, 37J20
\end{abstract}


\tableofcontents

\section{Introduction}\label{secc}
In the beginning of the 20th century G.G.\,Appelrot (H.\,Appelroth), in a series of his publications, gave deep and profound investigations of the motions of the Kowalevski top \cite{Kowa}. In \cite{App3,App4}, he found some specific families of motions and called them {\it especially remarkable motions}. These results were then included into Appelrot's large treatise on the gyroscopes theory and published under common title as \cite{Appel}. Appelrot defined especially remarkable motions as such solutions of the corresponding Euler\,--\,Poisson equations for which one of the separated Kowalevski variables remains constant. He pointed out four classes of such motions depending on various constraints on the integral constants leading to multiple roots of the Kowalevski polynomial involved in the separated equations. The corresponding solutions are expressed as elliptic functions of time. More thorough analytical and geometrical study of these motions was given by A.F.\,Ipatov \cite{Ipat}. Later it appeared that the Appelrot classes play the main role in the topological analysis of the Kowalevski case; the union of all trajectories of the especially remarkable motions in the phase space form a critical set of the integral map determining bifurcations of the integral manifolds \cite{KhPMM} (the rough topology of the Liouville foliation) and exact topological invariants of the problem \cite{BRF}.

The Appelrot classes received many generalizations in various directions. In \cite{PVLect} a new integrable case was found for a heavy dynamically symmetric gyrostat (a rigid body with a rotor characterized by constant gyrostatic momentum ${\boldsymbol \lambda}$); this case includes the Bobylev\,--\,Steklov case and the 4th Appelrot class as partial problems under some restrictions on the parameters. Another two exact solutions for a gyrostat with the Kowalevski type inertia tensor were found in \cite{EIPVHDan, PVMtt71}. These solutions cover the 2nd and the 3rd Appelrot classes as ${\boldsymbol \lambda}$ tends to zero. After the Kowalevski type heavy gyrostat was proved to be completely integrable \cite{Yeh}, it appeared that the motions pointed out in \cite{PVLect,EIPVHDan, PVMtt71} form a critical set of the Kowalevski\,--\,Yehia gyrostat \cite{Ryab1}. Other analogues of the Appelrot classes are associated with the problem of the rigid body motion in a \textit{double force field} (sometimes referred to as \textit{two constant fields}). This field could be, for example, a superposition of the gravity field and the constant magnetic field affecting a heavy magnetized rigid body with a fixed point. The equations of motion were given by O.I.\,Bogoyavlensky in \cite{Bogo} and presented as the Euler equation on the Lie co-algebra $e(3,2)^*$. In the same paper a partial integrable case (the subsystem with two degrees of freedom) was found generalizing the 1st Appelrot class.
More general system including a double field and a non-zero gyrostatic momentum was proved integrable in \cite{ReySem}. In this system, the generalizations of the 2nd, 3rd and 4th Appelrot classes with ${\boldsymbol \lambda}=0$ were discovered in \cite{Kh32, KhRCD}. For non-zero gyrostatic momentum, analogues of the Appelrot classes were found in \cite{KhHMJ}. In the general case these analogues also form a critical set of the integrable map; the critical trajectories are organized in almost everywhere Hamiltonian systems with less than three degrees of freedom (so-called critical subsystems). The knowledge of critical subsystems made it possible to obtain many analytical results (e.g, new separation of variables \cite{KhSavUMBeng,KhRCD09}, explicit descriptions of periodic motions \cite{KhJPA,KhIISmir40}) and to fulfill a number of topological investigations \cite{KhSavUMBeng,KhRCD09,KhND11}. On this basis, the complete description of the topology of the Kowalevski top in a double field in terms of net topological invariants was given \cite{RyabKhMatsbEn,KhRyabDME}.

In 2002, new integrable problems for the Kowalevski type gyrostat were found by V.V.\,Sokolov for the case of one axially symmetric field \cite{SokVV} and by V.V.\,Sokolov and A.V.\,Tsiganov for the case of a double field \cite{SokTsiE}. The last one also generalizes the Kowalevski type integrable gyrostat found by A.G.\,Reyman and M.A.\,Semenov-Tian-Shansky \cite{ReySem}. The investigation of the phase topology of the Sokolov\,--\,Tsiganov case was started in \cite{RyabTMF}.

In the present paper we obtain the generalized Appelrot classes for the case found in \cite{SokTsiE} and establish some properties of the appearing invariant manifolds (critical subsystems with two degrees of freedom). For each subsystem we point out the invariant relations defining its manifold and give a pair of independent first integrals. In terms of these integrals, using the invariant relations we explicitly calculate the types of the corresponding critical points of the integral map. It gives an analytical basis for the global topological study of the whole system with three degrees of freedom.

\section{Preliminaries}\label{sec0}
Let $\mathbf{v}$ denote the Lie algebra $e(3,2)=so(3)\oplus(\bbR^3 \dot{+}\bbR^3)$, the semi-direct sum of $so(3)$ and two copies of $\bbR^3$. The Euler equations on the adjacent space $\mathbf{v}^*\cong\bbR^9(\mathbf{M},\bda,\bdb)$ are (see \cite{Bogo})
\begin{equation}\label{eq2_1}
\begin{array}{l}
\displaystyle \frac{\rmd {\mathbf{M}}}{\rmd t} =\mathbf{M} {\times} \frac{\partial H} {\partial \mathbf{M}}  +  \bda{\times} \frac{\partial H} {\partial \bda}+ \bdb {\times} \frac{\partial H} {\partial \bdb},\\[3mm]
\displaystyle \frac {\rmd \bda}{\rmd t} = \bda {\times} \frac{\partial H} {\partial \mathbf{M}}, \qquad \frac {\rmd \bdb}{\rmd t}=
\bdb {\times} \frac{\partial H} {\partial \mathbf{M}}.
\end{array}
\end{equation}
These equations have a Hamiltonian form (with the Hamilton function $H$) with respect to the Lie\,--\,Poisson brackets on $\mathbf{v}^*$, i.e., for any function $f$ on $\mathbf{v}^*$ we have ${\rmd f}/{\rmd t}=\{H,f\}$. If
\begin{equation}\label{eq2_2}
  H=H^{(2)}+H^{(1)}+H^{(0)},
\end{equation}
where $H^{(i)}$ is a homogeneous degree $i$ function of the components $M_1,M_2,M_3$, then equations \eqref{eq2_1} describe some (generalized) problem of the rigid body dynamics, in which $H^{(2)}$ stands for the kinetic energy, $H^{(0)}=H^{(0)}(\bda, \bdb)$ is the potential energy.

The vector $\bdo =\partial H / \partial {\mathbf{M}}$ is called the angular velocity. The term $H^{(1)}$ in \eqref{eq2_1} generates forces which are usually called gyroscopic. The canonical unit bases in $so(3)\cong \bbR^3 = \{\mathbf{M}\}$, $\bbR^3=\{\bda\}$ and $\bbR^3=\{\bdb\}$ are identified with some orthonormal basis $O\mathbf{e}_1\mathbf{e}_2\mathbf{e}_3$ fixed in the rotating body (the moving frame with the origin $O$ at the fixed point of the body).

For a generic case we suppose that $\bda{\times}\bdb \ne 0$. The phase space of equations \eqref{eq2_1} is a \mbox{6-di}mensional submanifold $\mP^6$ in $\mathbf{v}^*$ defined as a common level of the Casimir functions (geometric integrals)
\begin{equation*}
  \bda {\bdot} \bda = a^2, \quad \bdb {\bdot} \bdb = b^2, \quad \bda {\bdot} \bdb = c \qquad (a \geqslant b > 0, \; |c| < a b).
\end{equation*}
On $\mP^6$ there exists a natural symplectic structure compatible with the restriction of the Lie\,--\,Poisson brackets.

If $H^{(0)}$ is a linear function in $\bda, \bdb$ and essentially depends on both vectors, corresponding potential field is called a double field. In this case the vectors $\bda, \bdb$ are called the fields intensities and, obviously, are constant ``in space'' as shown by the second group of equations \eqref{eq2_1} known as the Poisson equations. For a double field we write
\begin{equation}\label{eq2_3}
  H^{(0)} = - \ve_0 (\mathbf{r}_1  {\bdot} \bda + \mathbf{r}_2  {\bdot} \bdb)+{\rm const}
\end{equation}
with constant vectors $\mathbf{r}_1,\mathbf{r}_2$. These vectors are pointing from $O$ to the centers of application of the two constant fields.  The redundant parameter $\ve_0 \geqslant 0$ allows the essential part of $H^{(0)}$ to vanish while $a^2+b^2 \ne 0$. In what follows, for the cases with $H^{(1)}$ independent of $\bda,\bdb$ and non-constant $H^{(0)}$ we always put $\ve_0 =1$.

If $\mathbf{r}_1 {\times} \mathbf{r}_2\ne 0$, then $\mathbf{r}_1, \mathbf{r}_2$ can be chosen to form an orthonormal pair and therefore can be taken as the first two vectors of the moving frame
\begin{equation}\label{eq2_4}
 \mathbf{r}_1= \mathbf{e}_1, \qquad \mathbf{r}_2=\mathbf{e}_2.
\end{equation}
Note that for a large class of Hamilton functions including all considered here, without loss of generality one can suppose that $\bda$ and $\bdb$ are orthogonal \cite{KhRCD,KhHMJ}.
This fact will be used later in the change of variables which simplifies calculations and formulas.

We say that the system \eqref{eq2_1} with the Hamilton function \eqref{eq2_2} has the Kowalevski configuration if in addition to \eqref{eq2_3}, \eqref{eq2_4} we have
\begin{equation*}
\begin{array}{l}
  H^{(2)}={\displaystyle \frac{1}{2}} \mathbf{g} \mathbf{M} {\bdot} \mathbf{M}, \qquad \mathbf{g}=\mathop{\rm diag}\nolimits \{\frac{1}{2},\frac{1}{2},1\}.
\end{array}
\end{equation*}
Here the constant multiplier of $H^{(2)}$ is due only to the choice of measure units. The essential condition is the ratio $1{:}1{:}2$ of the eigenvalues of the matrix $\mathbf{g}$ inverse to the inertia tensor of the body.

The classical Kowalevski case corresponds to $H^{(1)}\equiv 0, \bdb\equiv 0$:
\begin{equation}\label{eq2_5}
  H_{\rm class}=\omega_1^2+\omega_2^2 +\frac{1}{2}\omega_3^2-\alpha_1.
\end{equation}
The first integrals then are the area integral
\begin{equation}\label{eq2_6}
\begin{array}{l}
  L_{\rm class} = {\displaystyle \frac{1}{2}} \mathbf{M} {\bdot} \bda   = \alpha_1 \omega_1 + \alpha_2 \omega_2+ {\displaystyle \frac{1}{2}} \alpha_3 \omega_3\\
\end{array}
\end{equation}
and the famous Kowalevski integral
\begin{equation}\label{eq2_7}
\begin{array}{l}
  K_{\rm class} = (\omega_1^2-\omega_2^2+\alpha_1)^2+ (2\omega_1 \omega_2+\alpha_2)^2.
\end{array}
\end{equation}

For the general Kowalevski configuration under certain conditions the first integrals additional to $H$ were found in \cite{Bogo,Yeh,ReySem}. O.I.\,Bogoyavlensky \cite{Bogo} considered the case $H^{(1)}\equiv 0, \bdb \ne 0$ and found the first integral
\begin{equation*}
\begin{array}{l}
  K_{\rm B} = (\omega_1^2-\omega_2^2+\alpha_1-\beta_2)^2+ (2\omega_1 \omega_2+\alpha_2+\beta_1)^2.
\end{array}
\end{equation*}
Since for $\bdb \ne 0$ there is no area integral, Bogoyavlensky pointed out the special case $K_{\rm B}=0$ and proved the integrability of the system on this 4-dimensional submanifold in $\mP^6$. In more general case
\begin{equation}\label{eq2_8}
  H^{(1)} = -\lambda M_3, \qquad \bdb \ne 0 \qquad (\lambda={\rm const})
\end{equation}
H.M.\,Yehia \cite{Yeh} found the integral
\begin{equation*}
\begin{array}{l}
K_{\rm Y} = (\omega _1^2  - \omega _2^2  + \alpha _1  - \beta _2 )^2 +
(2\omega _1 \omega _2  + \alpha _2  + \beta _1 )^2 \\[2mm]
\phantom{K = (}+ 2\lambda[(\omega_3-\lambda) ( \omega_1^2+ \omega_2
^ 2)
+ 2 \omega_1 \alpha_3 + 2 \omega_2 \beta_3 ].
\end{array}
\end{equation*}
Here the constant vector of the gyrostatic momentum is ${\boldsymbol \lambda} = \lambda \mathbf{e}_3$. For simplicity, the term ``gyrostatic momentum'' in this case is used for the scalar parameter $\lambda$. If in addition we suppose $\bdb \equiv 0$, then there exists the general first integral
\begin{equation*}
\begin{array}{l}
  L_{\rm Y} =  \alpha_1 \omega_1 + \alpha_2 \omega_2+ {\displaystyle \frac{1}{2}} \alpha_3 (\omega_3+\lambda)
\end{array}
\end{equation*}
similar to \eqref{eq2_6} and the system becomes completely integrable. This case is known as the Kowaslevski\,--\,Yehia gyrostat.

In \cite{ReySem} a Lax representation of equations \eqref{eq2_1} in the case \eqref{eq2_8} was found that gave rise to the general first integral
\begin{equation*}
\begin{array}{l}
G_{\rm RS}= \left[\alpha_1 \omega_1+\alpha_2 \omega_2+\displaystyle \frac{1}{2}\alpha_3 (\omega_3+\lambda)\right]^2+ \left[\beta_1 \omega_1+\beta_2 \omega_2+\displaystyle \frac{1}{2}\beta_3 (\omega_3+\lambda)\right]^2 \\
\phantom{G_{\rm RS}} +(\omega_3-\lambda) \left[(\alpha_2 \beta_3-\alpha_3 \beta_2) \omega_1+(\alpha_3 \beta_1-\alpha_1 \beta_3) \omega_2+\displaystyle \frac{1}{2} (\alpha_1 \beta_2-\alpha_2 \beta_1) (\omega_3+\lambda)\right] \\
\phantom{G_{\rm RS}}  -\bdb^2 \alpha _1  - \bda^2 \beta_2 + \bda {\bdot} \bdb (\alpha_2+\beta_1).
\end{array}
\end{equation*}
Thus, the system \eqref{eq2_1} with $H=H_\lambda$, where
\begin{equation}\label{eq2_9}
  H_\lambda = \frac{1}{4}(M_1^2+M_2^2)+\frac{1}{2}M_3^2 - \lambda M_3 -\alpha_1-\beta_2 + \frac{\lambda^2}{2},
\end{equation}
having three first integrals $H_\lambda, K_{\rm Y}, G_{\rm RS}$ in involution is a Liouville integrable system with three degrees of freedom. It is known as the Kowalevski \textit{gyrostat} in a double field. The relation between $\mathbf{M}$ and $\bdo$ then is
\begin{equation*}
  M_1=2 \omega_1, \quad M_2=2 \omega_2, \quad M_3=\omega_3+\lambda.
\end{equation*}
The constant term $\frac{\lambda^2}{2}$ in $H_\lambda^{(0)}$ is introduced to make $H_\lambda$ the total energy of the gyrostat. If $\lambda=0$ the system is called the Kowalevski \textit{top} in a double field. Obviously, for $\bdb=0$ we have $G_{\rm RS}=L_{\rm Y}^2$, so the new integral generalizes the area type integral for a double field.

The Kowalevski top and the Kowalevski gyrostat in a double field up to this moment have not been reduced to quadratures. However, the phase topology of the Kowalevski top in a double field is completely investigated. The collection of results and the relevant references can be found in \cite{RyabKhMatsbEn, KhRyabDME}. The Kowalevski gyrostat in a double field is also actively studied \cite{KhHMJ,KhIISmir40}. These investigations are based on the notion of a critical subsystem, i.e., an even-dimensional invariant submanifold in the phase space consisting of critical points of the integral map. The idea of a critical subsystem is as follows \cite{KhRCD, KhHMJ}. Let
\begin{equation*}
  J=H{\times}K{\times}G: \mP^6 \to \bbR^3
\end{equation*}
be the integral map of the system \eqref{eq2_1} and let $\mathcal{C}$ be the set of its critical points. Obviously, $\mathcal{C}$ is stratified by $\rk J$ and for $x\in \mathcal{C}$ the number $\rk J(x)$ is called the rank of the critical point $x$. The set $J(\mathcal{C})$ is called the bifurcation diagram of $J$ and is also stratified by the maximal rank on the pre-images of its points. Smooth \mbox{2-d}imensional surfaces or isolated \mbox{1-d}imensional curves in $J(\mathcal{C})$ are called the leaves of the bifurcation diagram.
Suppose that
\begin{equation}\label{eq2_10}
\mL(h,k,g)=0
\end{equation}
is the equation of a 2-surface $\Pi_\mL\subset \bbR^3$ bearing one of the \mbox{2-d}imensional leaves. Let
\begin{equation}\label{eq2_11}
\Phi_\mL=\mL \circ J: \mP^6 \to \bbR.
\end{equation}
Then the critical subsystem $\mathcal{J}_\mL$ (of rank 2) is the closure of the set of critical points of rank~2 belonging to the level $\Phi_\mL=0$ (to the pre-image of $\Pi_\mL$). Obviously, $\mathcal{J}_\mL$ is an invariant subset in $\mP^6$ consisting of critical points of $J$. Therefore, speaking of a critical subsystem we always have in mind the induced dynamics. The subset $\mathcal{J}_\mL$ is almost everywhere a 4-dimensional manifold given by the equations
\begin{equation}\label{eq2_12}
\Phi_\mL=0, \qquad \rmd \Phi_\mL =0
\end{equation}
and the induced dynamical system on it is almost everywhere Hamiltonian with two degrees of freedom. The global Hamiltonian structure may have singularities at the points where either $\mathcal{J}_\mL$ fails to be smooth or the 2-form induced on $\mathcal{J}_\mL$ by the symplectic structure degenerates. Critical subsystems of rank 1 (with one degree of freedom) can be obtained in a similar way starting with two equations of the type~\eqref{eq2_10}.

In analytical case (particularly, in the case when all functions and equations have polynomial structure) critical subsystems are defined uniquely if we consider all irreducible equations of the type \eqref{eq2_10}. Such systems generate a stratification of the phase space and determine bifurcations that occur to the Liouville tori. Invariant relations describing a critical subsystem make it possible to explicitly calculate the types of critical points for all ranks.

Critical subsystems were found for the Kowalevski top in a double field in \cite{Bogo,Kh32,KhRCD} and for the Kowalevski gyrostat in a double field in \cite{KhJPA, KhHMJ}. In both cases, one of these subsystems has one degree of freedom and is formed completely by the pendulum type motions.
For the top, there exist three subsystems with two degrees of freedom. For the gyrostat, one of the subsystems (the Bogoyavlensky case) splits, the parts of the new manifold join the other subsystems. In two of three critical subsystems for the top in a double field the algebraic separation of variables was obtained \cite{KhSavUMBeng,KhRCD09}. This turns the topological investigation into a strictly analytical procedure~\cite{KhND11}.

Still, further generalization of the problem \eqref{eq2_9} was found in \cite{SokTsiE}. To describe this case take the system with the Kowalevski configuration and put
\begin{equation*}
  \displaystyle H^{(0)} = - \ve_0 ({\mathbf{e}}_1  {\bdot} \bda + {\mathbf{e}}_2  {\bdot} \bdb)+\frac{\lambda^2}{2}, \quad   H^{(1)} = - \lambda M_3 - \ve_1  \mathbf{M} {\bdot} ({\mathbf{e}}_1  {\times} \bda + {\mathbf{e}}_2  {\times} \bdb).
\end{equation*}
The full Hamilton function takes the form
\begin{equation}\label{eq2_13}
\begin{array}{l}
  \displaystyle H = \frac{1}{4}(M_1^2+M_2^2)+\frac{1}{2}(M_3-\lambda)^2 \\[3mm]
 \displaystyle \quad - \ve_1  [ (\alpha_2 M_3 - \alpha_3 M_2 ) + (\beta_3 M_1 - \beta_1 M_3 )] -\ve_0 (\alpha_1+\beta_2).
\end{array}
\end{equation}
For this function, the Lax representation with a spectral parameter was given in \cite{SokTsiE} proving the complete Liouville integrability.
The authors of \cite{SokTsiE} called their case the generalized two fields gyrostat. In what follows we call the system \eqref{eq2_1} with the Hamilton function \eqref{eq2_13} the GTFG  system. P.E.\,Ryabov in \cite{RyabTMF} started the topological analysis of the GTFG system. First of all, he obtains the explicit formulas of the commutating integrals $K$ and $G$ generalizing $K_{\rm Y}$ and $G_{\rm RS}$, thus expressing in terms of these integrals the coefficients of the algebraic curve associated with the Lax representation:
\begin{equation*}
\begin{array}{l}
\displaystyle K= \left[\frac{1}{4}(M_1^2-M_2^2)+\ve_0 (\alpha_1-\beta_2)+\ve_1 [(\bda{\times}{\mathbf{M}}){\bdot}\mathbf{e}_1-(\bdb{\times}{\mathbf{M}}){\bdot}\mathbf{e}_2-\ve_1 (\bda^2-\bdb^2)]     \right]^2\\[3mm]
\displaystyle \quad +\left[\frac{1}{2}M_1 M_2+\ve_0 (\alpha_2+\beta_1)+\ve_1 [(\bda{\times}{\mathbf{M}}){\bdot}\mathbf{e}_2+(\bdb{\times}{\mathbf{M}}){\bdot}\mathbf{e}_1-2\ve_1  \bda {\bdot}\bdb]    \right]^2\\[3mm]
\displaystyle \quad +2\lambda \left[\frac{1}{4}(M_1^2+M_2^2)(M_3-2\lambda)+\ve_0 (\alpha_3 M_1+\beta_3 M_2) \right] \\[3mm]
\displaystyle \quad -2 \ve_1  \lambda \left[ \alpha_2 M_1^2-\beta_1 M_2^2-(\alpha_1-\beta_2)M_1 M_2+\ve_1 (\bda^2+\bdb^2)-2 \ve_1 (\bda{\times}\bdb){\bdot} \mathbf{M}\right], \\[3mm]
\displaystyle G= \frac{1}{4}\left[(\mathbf{M} {\bdot}\bda)^2+ (\mathbf{M} {\bdot}\bdb)^2\right]+\frac{1}{2}(M_3-2\lambda)(\bda{\times}\bdb) {\bdot}\mathbf{M} \\[3mm]
\displaystyle \quad +\ve_0 \left[\bdb^2\alpha_1 + \bda^2 \beta_2 -(\bda{\bdot}\bdb)(\alpha_2+\beta_1)\right]\\[3mm]
\displaystyle \quad -\ve_1  \left[\bdb^2 (\bda{\times}\mathbf{M}){\bdot}\mathbf{e}_1+ \bda^2 (\bdb{\times}\mathbf{M}){\bdot}\mathbf{e}_2-(\bda{\bdot}\bdb)[(\bda{\times}\mathbf{M}){\bdot}\mathbf{e}_2+(\bdb{\times}\mathbf{M}){\bdot}\mathbf{e}_1 ]               \right]  \\[3mm]
\displaystyle \quad = \frac{1}{4}[(\mathbf{M} {\bdot}\bda)^2+ (\mathbf{M} {\bdot}\bdb)^2]+\frac{1}{2}(M_3-2\lambda)(\bda{\times}\bdb) {\bdot}\mathbf{M} \\[3mm]
\displaystyle \quad +\ve_0 (\bda{\times}\bdb){\bdot} \left[\mathbf{e}_2{\times}\bda+\mathbf{e}_1{\times}\bdb \right]-\ve_1  \mathbf{M} {\bdot} \left[\mathbf{e}_1{\times}\bigl(\bdb{\times}(\bda{\times}\bdb)\bigr)+ \mathbf{e}_2{\times}\bigl(\bda{\times}(\bda{\times}\bdb)\bigr) \right].
\end{array}
\end{equation*}

Next, in the work \cite{RyabTMF} for the case $\lambda=0$ four critical subsystems of rank 2 are pointed out. For these subsystems, the constraints on the integral constants are derived from the condition that the algebraic curve of the Lax pair has a singular point. For two subsystems the invariant relations in the phase space are found. Another two are defined implicitly by the equations of the type \eqref{eq2_10}. This form of description does not provide a clear way to calculate the types of corresponding critical points. In what follows, we give explicit equations of three critical subsystems generalizing the Appelrot classes of motion.

To present formulas in a shorter form, let us introduce the linear change of variables ($\rmi^2=-1$)
\begin{equation}\label{eq2_14}
\begin{array}{c}
\begin{array}{ll}
x_1 = (\alpha_1  - \beta_2) + \rmi (\alpha_2  + \beta_1),&
x_2 = (\alpha_1  - \beta_2) - \rmi(\alpha_2  + \beta_1 ), \\
y_1 = (\alpha_1  + \beta_2) + \rmi(\alpha_2  - \beta_1), & y_2 =
(\alpha_1  + \beta_2) -
\rmi(\alpha_2  - \beta_1), \\
z_1 = \alpha_3  + \rmi\beta_3, &
z_2 = \alpha_3  - \rmi\beta_3,
\end{array}\\
w_1 = \frac{1}{2}(M_1  + \rmi M_2) , \quad w_2 = \frac{1}{2}(M_1  - \rmi M_2), \quad w_3 = M_3-\lambda.
\end{array}
\end{equation}

For convenience, we use the Poisson brackets obtained from the initial Lie\,--\,Poisson brackets dividing by $\rmi$. Then for the phase vector $\mathbf{u} = (w_1,w_2,w_3,x_1,x_2,y_1,y_2,z_1,z_2)$
the matrix of the basic brackets $\{u_i,u_j\}$ is
\begin{equation*}
   \left(\begin{array}{ccccccccc}
  0 &  -\frac{1}{2}(w_3 + \lambda) &  w_1 &  0 &  -z_2 &  0 &  -z_1 &  \frac{1}{2}x_1 &  \frac{1}{2}y_1\\
  \frac{1}{2}(w_3 + \lambda) &  0 &  -w_2 &  z_1 &  0 &  z_2 &  0 &  -\frac{1}{2}y_2 &  -\frac{1}{2}x_2\\
  -w_1 &   w_2 &  0 &  -x_1 &  x_2 &  -y_1 &  y_2 &  0 &  0\\
  0 &  -z_1 &  x_1 &  0 &  0 &  0 &  0 &  0 &  0\\
  z_2 &  0 &  -x_2 &  0 &  0 &  0 &  0 &  0 &  0\\
  0 &  -z_2 &  y_1 &  0 &  0 &  0 &  0 &  0 &  0\\
  z_1 &  0 &  -y_2 &  0 &  0 &  0 &  0 &  0 &  0\\
  -\frac{1}{2}x_1 &  \frac{1}{2}y_2 &  0 &  0 &  0 &  0 &  0 &  0 & 0\\
  -\frac{1}{2}y_1 &  \frac{1}{2}x_2 &  0 &  0 &  0 &  0 &  0 &  0 &  0
  \end{array}\right).
\end{equation*}
Using variables \eqref{eq2_14} we suppose that the intensities of the forces are made orthogonal according to the procedure described in \cite{KhRCD,KhHMJ}.
It is known that if two orthogonal fields $\bda,\bdb$ ($\bda {\bdot} \bdb = 0$) have equal intensities $a=b$, then there exists a symmetry group and the system can be reduced to two degrees of freedom \cite{Yeh,KhIISavJPA}. Therefore, for the general case we take $a \ne b$ and choose the notation in such a way that $a > b > 0$. Then the equations of the phase space $\mP^6$ (the geometric integrals) take the form
\begin{equation}\label{eq2_15}
  z_1^2+x_1 y_2 =r^2, \quad z_2^2+x_2 y_1 =r^2, \quad x_1 x_2 + y_1 y_2 +2 z_1 z_2 = 2 p^2,
\end{equation}
where $p^2=a^2+b^2$ and $r^2=a^2-b^2$ ($p>r>0$). Since the components of $\mathbf{M}, \bda, \bdb$ are real, equations \eqref{eq2_14} define in $\mathbb{C}^9$ the subspace $V^9$. The latter is \mbox{9-dim}ensional over $\bbR$. In turn, the phase space $\mP^6$ is defined in $V^9$ by equations \eqref{eq2_15}.

Let $f$ be an arbitrary function on $V^9$. For brevity, the term ``critical point of $f$'' will always mean a critical point of the
restriction of $f$ to $\mP^6$. Similarly, $\rmd f$ means the restriction of the differential of $f$ to the set of vectors tangent to $\mP^6$.

While calculating critical points of various functions (in the above sense), it is convenient to avoid introducing Lagrange multipliers
for the restrictions \eqref{eq2_15}. Notice that the following vector fields
\begin{equation*}
\begin{array}{l}
X_1=\partial_{w_1}, \; X_2=\partial_{w_2},\;X_3=\partial_{w_3}, \\
Y_1= z_2 \partial_{x_2} + z_1 \partial_{y_2} -\frac{1}{2} x_1 \partial_{z_1} - \frac{1}{2} y_1 \partial_{z_2},\\
Y_2= z_1 \partial_{x_1} + z_2 \partial_{y_1}  - \frac{1}{2} y_2 \partial_{z_1}-\frac{1}{2} x_2 \partial_{z_2},\\
Y_3= x_1 \partial_{x_1} - x_2 \partial_{x_2} + y_1 \partial_{y_1} - y_2 \partial_{y_2}
\end{array}
\end{equation*}
are tangent to $\mP^6 \subset V^9$ and linearly independent at any point of $\mP^6$. Then the set of critical points of $f$ is defined by
the system of equations
\begin{equation}\label{eq2_16}
\begin{array}{lll}
X_1 f=0, & X_2 f=0, & X_3 f=0, \\
Y_1 f=0, & Y_2 f=0, & Y_3 f=0.
\end{array}
\end{equation}
This fact will be used below.

\section{The Appelrot classes}\label{sec2}
Consider the classical Kowalevski problem. According to the separated equations obtained by S.\,Kowalevski \cite{Kowa} and the constant shift in the variables introduced later by F.\,K\"{o}tter \cite{Kott}, the especially remarkable motions in Appelrot's definition \cite{Appel} are possible only if the polynomial in $s$
\begin{equation*}
S_{\rm Kow} (s)=[(s-h)^2 - k]\{s[(s-h)^2+a^2-k]-2\ell^2\}
\end{equation*}
has a multiple root. Here $h,\ell,k$ stand for the constants of the integrals \eqref{eq2_5} -- \eqref{eq2_7} respectively.
Appelrot points out the following classes of such motions.

The 1st class is characterized by the relation
\begin{equation}\label{eq3_1}
  k=0.
\end{equation}
It immediately follows from \eqref{eq2_7} that the corresponding invariant manifold is defined in the reduced phase space $\mP^5=\bbR^3(\bdo){\times}S^2(\bda)$ by the following two equations
\begin{equation}\label{eq3_2}
  \omega_1^2-\omega_2^2+\alpha_1=0, \qquad 2\omega_1 \omega_2+\alpha_2 =0 .
\end{equation}

The 2nd and the 3rd classes are generated by the relation
\begin{equation}\label{eq3_3}
  (2\ell^2-a^2h)^2-a^4 k =0.
\end{equation}
The difference between two classes is in the sign of $2\ell^2-a^2h$. This value is positive in the 2nd class and negative in the 3rd one. Zero value, obviously, gives a subset of motions from the 1st class treated usually as common motions for all three classes. In $\mP^5$ we have a connected manifold defined by two invariant relations (e.g. \cite{Kh32})
\begin{equation}\label{eq3_4}
\begin{array}{c}
\displaystyle{(\alpha _1^2  + \alpha _2^2 )\omega _3 - 2(\alpha _1
\omega _1  + \alpha _2 \omega _2 )\alpha _3  =
0,} \\
\displaystyle{2\alpha _1 \alpha _2 (\omega _1^2  - \omega _2^2  +
\alpha _1 ) - (\alpha _1^2  - \alpha _2^2 )(2\omega _1 \omega _2  +
\alpha _2 ) = 0}.
\end{array}
\end{equation}

The 4th Appelrot class is generated by the zero discriminant of the last multiplier in~$S_{\rm Kow}$
\begin{equation*}
  27\ell^4-2h(9a^2+h^2-9k)\ell^2+(a^2-k)(a^2+h^2-k)=0.
\end{equation*}
It is more convenient to write the equations of this surface taking $h$ and $s$ for the parameters:
\begin{equation}\label{eq3_5}
k=a^2+(h-s)(h-3s), \qquad \ell^2=s^2(h-s).
\end{equation}
Appelrot shows that all especially remarkable motions of the 4th class satisfy the following conditions
\begin{equation}\label{eq3_6}
s = -\frac{\ell}{\omega_1}={\rm const}, \qquad \omega_2=0.
\end{equation}
Since $\omega_2=0$ yields $\dot \omega_1=0$, to obtain the system of invariant relations we replace the first condition in \eqref{eq3_6} with $\dot \omega_2=0$. Thus, the equations of the corresponding manifold in $\mP^5$ are
\begin{equation}\label{eq3_7}
  \omega_1 \omega_3+\alpha_2 =0, \qquad \omega_2=0.
\end{equation}

In the following sections we point out the equations of three invariant manifolds generalizing the Appelrot classes for the Hamilton function \eqref{eq2_13}.
Two of them exist only for the top with $\lambda=0$ and correspond to the bifurcation surfaces found in \cite{RyabTMF}, the third one exists for arbitrary values of $\lambda$ and generalizes the cases shown in the works~\cite{KhRCD09,KhJPA,KhHMJ}.

\section{The extension of the 1st Appelrot class}\label{sec3}
It is known that the 1st Appelrot class is not stable with respect to the gyrostatic momentum. Indeed, already for the case of a single field (the gravity field) the corresponding invariant manifold disappears for the Kowalevski\,--\,Yehia gyrostat with $\lambda \ne 0$ (see \cite{Ryab1}). To illustrate this fact, put $\bdb=0, \ve_1 =0$ and consider the image of the first three Appelrot classes in the space of the integral constants $\bbR^3(\ell,h,k)$. It is a part of the bifurcation diagram of the corresponding integral map. The cross-section of this set by the plane $\ell={\rm const}$ is shown in Figure~\ref{fig1}. We see the reconstruction of the Appelrot classes to form new critical subsystems. Therefore, in this section for the general case \eqref{eq2_13} we suppose that $\lambda=0$.

\begin{figure}[!ht]
\centering
\includegraphics[width=0.7\textwidth, keepaspectratio = true]{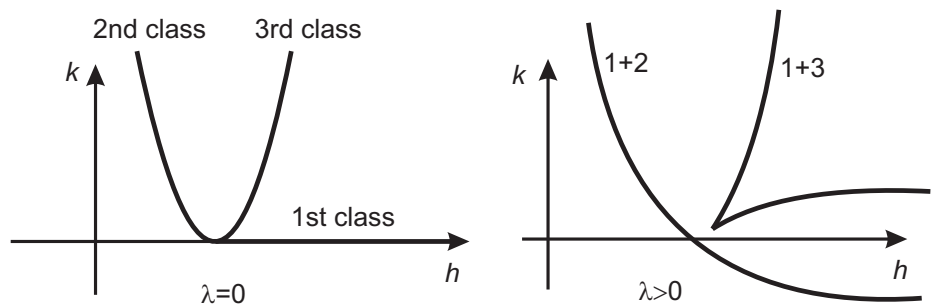}
\caption{Splitting of the Appelrot classes.}\label{fig1}
\end{figure}

\begin{theorem}[\cite{RyabTMF}]\label{theo1}
The GTFG  system with $\lambda=0$ has the invariant manifold $\mm_1$ defined by the system of invariant relations
\begin{equation}\label{eq4_1}
  Z_1=0, \qquad Z_2=0,
\end{equation}
where
\begin{equation*}
\begin{array}{l}
  Z_1= \frac{1}{4}(M_1^2-M_2^2)+\ve_0 (\alpha_1-\beta_2)+\ve_1 [(\bda{\times}{\mathbf{M}}){\bdot}\mathbf{e}_1-(\bdb{\times}{\mathbf{M}}){\bdot}\mathbf{e}_2] \\[3mm]
  \phantom{ Z_1} - \ve_1 ^2(\bda^2-\bdb^2),     \\[3mm]
  Z_2=\frac{1}{2}M_1 M_2+\ve_0 (\alpha_2+\beta_1)+\ve_1 [(\bda{\times}{\mathbf{M}}){\bdot}\mathbf{e}_2+(\bdb{\times}{\mathbf{M}}){\bdot}\mathbf{e}_1-2\ve_1  \bda {\bdot}\bdb].
\end{array}
\end{equation*}
The induced system on $\mm_1$ is the critical subsystem generated by the conditions $\eqref{eq2_12}$ with
\begin{equation}\label{eq4_2}
  \mL=\mL_1= k.
\end{equation}
\end{theorem}

The proof follows immediately from the fact that in this case
\begin{equation}\label{eq4_3}
  K=Z_1^2+Z_2^2.
\end{equation}
Obviously, for the classical Kowalevski case we come to equations \eqref{eq3_1}, \eqref{eq3_2}, thus obtaining the generalization of the first Appelrot class of motions. If $\bdb\ne 0$, putting $\ve_0 =1, \ve_1 =0$ we come to the partial integrable case found by Bogoyavlensky~\cite{Bogo}.

Let us point out some features of the critical subsystem $\mm_1$.

\begin{prop}\label{prop1}
The set $\mm_1$ is a smooth 4-dimensional manifold except for the values of the parameters satisfying one of the relations
\begin{equation}\label{eq4_4}
  \ve_0 =(a \pm b) \ve_1 ^2,
\end{equation}
when $\mm_1$ fails to be smooth in the neighborhood of the following equilibria
\begin{equation}\label{eq4_5}
   \alpha_1 = a, \quad \alpha_2=\alpha_3=0, \quad \beta_1=\beta_3=0, \quad \beta_2 = \pm b,\quad M_1=M_2=M_3=0.
\end{equation}
\end{prop}

\begin{proof}
In variables \eqref{eq2_14} we have $K=\xi_1 \xi_2$, where
\begin{equation}\label{eq4_6}
  \begin{array}{l}
    \xi_1=w_1^2 + \ve_0 x_1 - \rmi \ve_1 (x_1 w_3 - 2 z_1 w_1) - \ve_1 ^2 r^2, \\
    \xi_2=w_2^2 + \ve_0 x_2 + \rmi \ve_1 (x_2 w_3 - 2 z_2 w_2) - \ve_1 ^2 r^2.
  \end{array}
\end{equation}
The system \eqref{eq4_1} is equivalent to
\begin{equation}\label{eq4_7}
  \xi_1=0, \qquad \xi_2=0.
\end{equation}
To find possible singularities of $\mm_1$ we apply equations \eqref{eq2_16} to the function with undefined multipliers $f=\mu_1 \xi_1+\mu_2 \xi_2$ ($\mu_2=\overline{\mu_1}$, $\mu_1 \mu_2 \ne 0$).
The only solution is
\begin{equation*}
w_1=w_2=w_3=0, \quad x_1=x_2=\frac{\ve_1 ^2 r^2}{\ve_0 }, \quad y_1=y_2=\frac{\ve_0 }{\ve_1 ^2}, \quad z_1=z_2=0,
\end{equation*}
and equations \eqref{eq2_15} reduce to \eqref{eq4_4}. According to the choice of the sign in \eqref{eq4_4} we come to the corresponding point \eqref{eq4_5}.
\end{proof}

It follows from \eqref{eq4_3} that the Poisson bracket of $Z_1$ and $Z_2$ is a partial integral on $\mm_1$. In the classical case such integral was found by D.N.\,Goryachev, and in the Bogoyavlensky case it has the form \cite{Bogo}
\begin{equation*}
  F_{\rm B}= \frac{1}{4} (M_1^2 + M_2^2) M_ 3+ \alpha_3 M_1 + \beta_3 M_2.
\end{equation*}

Recall one general fact from the symplectic geometry.
\begin{lemma}\label{lem1}
Suppose a submanifold $\mm$ of a symplectic manifold $\mathcal{N}$ is defined by the system of independent equations
\begin{equation}\label{eq4_8}
  f_1=0, \qquad f_2=0.
\end{equation}
Then the 2-form on $\mm$ induced by the symplectic structure of $\mathcal{N}$ degenerates exactly on the set $\{f_1,f_2\}=0$.
\end{lemma}

Since critical subsystems are usually described by the systems of the type \eqref{eq4_8}, the induced symplectic structure may degenerate on a set of codimension 1. In this case the
subsystem is said to be almost Hamiltonian.

The next statement follows from the results of \cite{RyabTMF} with the above agreement on the Poisson brackets.

\begin{prop}\label{prop2}
The restriction to $\mm_1$ of the Poisson bracket $\{\xi_1,\xi_2\}$ is a partial integral of the induced system
\begin{equation}\label{eq4_9}
  \{\xi_1,\xi_2\}= -2 \,N \, F,
\end{equation}
where $N$ and $F$ are both partial integrals and can be presented as follows
\begin{equation}\label{eq4_10}
   N = \sqrt{\bigl[ \ve_0 + \rmi \ve_1  (w_3+\rmi \ve_1  y_1) \bigr] \bigl[ \ve_0 - \rmi \ve_1  (w_3-\rmi \ve_1  y_2) \bigr]}+\ve_1 ^2 \sqrt{\mathstrut x_1 x_2 },
\end{equation}
\begin{equation}\label{eq4_11}
   F = \sqrt{\mathstrut x_1 x_2 } w_3-\frac{x_2 z_1 w_1+x_1 z_2 w_2}{\sqrt{\mathstrut x_1 x_2 }}+\rmi \ve_1  r^2\frac{x_1-x_2}{\sqrt{\mathstrut x_1 x_2 }}.
\end{equation}
The connection of the general integrals $H,G$ on $\mm_1$ with these new integrals is
\begin{equation*}
  p^2 H -2 G+\ve_1 ^2 r^4 = \frac{F^2}{2}, \qquad \ve_0 ^2+2 \ve_1 ^2(H+\ve_1 ^2 p^2) = N^2.
\end{equation*}
For small enough values of $\ve_1 $ the integral $N$ does not vanish on $\mm_1$. The 2-form induced on $\mm_1$ by the symplectic structure of $\mP^6$ is non-degenerate everywhere except for the set
\begin{equation}\label{eq4_12}
  N \, F=0.
\end{equation}
\end{prop}

\begin{proof} Expression \eqref{eq4_9} is obtained by a straightforward calculation using \eqref{eq4_7}. We see that if $\ve_1 =0, \ve_0 >0$, then $N=\ve_0 \ne 0$. For $\ve_1 \ne 0$ the equation $N=0$ yields
\begin{equation*}
  \begin{array}{c}
  \displaystyle x_1=x_2=0, \quad y_1=\frac{\ve_0 +\rmi \ve_1  w_3}{\ve_1 ^2},\quad y_1=\frac{\ve_0 +\rmi \ve_1  w_3}{\ve_1 ^2},\quad y_2=\frac{\ve_0 -\rmi \ve_1  w_3}{\ve_1 ^2}, \\
  z_1^2=z_2^2=r^2, \quad w_1 =-\rmi \ve_1  z_1, \quad w_2=\rmi \ve_1  z_2.
\end{array}
\end{equation*}
Then from \eqref{eq2_15} we obtain ${\ve_0 ^2}-2(p^2-r^2){\ve_1 ^4}+{w_3^2} {\ve_1 ^2}=0$.
Therefore, $N$ can vanish only for the values
\begin{equation*}
  \ve_1 ^4 \geqslant \frac{\ve_0 ^2}{2(p^2-r^2)}.
\end{equation*}
The property of the induced 2-form follows from Lemma~\ref{lem1}.
\end{proof}


\section{The extension of the 2nd and 3rd Appelrot classes}\label{sec4}
As one can see from Figure~\ref{fig1}, the 2nd and 3rd Appelrot classes forming one invariant manifold in the classical Kowalevski case split and glue with other sets of critical points when $\lambda \ne 0$. For the gyrostat in a double field ($\ve_1 =0$) new subsystems are described in \cite{KhHMJ}. In this section we take $\ve_1 \ne 0$, but still suppose that $\lambda=0$.

Let us calculate the derivative in virtue of the system \eqref{eq2_1} of the function $F$ which serves as a partial integral for the extension of the 1st Appelrot class and is defined by equation \eqref{eq4_11}:
\begin{equation}\label{eq5_1}
  F'=\{H,F\}=\frac{r^2}{2 \sqrt{x_1 x_2}} F_1.
\end{equation}
Here prime stands for $\rmd/\rmd(\rmi t)$, and with the notation \eqref{eq4_6}
\begin{equation*}
  F_1= \frac{x_2}{x_1}\xi_1-\frac{x_1}{x_2}\xi_2.
\end{equation*}
Due to \eqref{eq4_7} on $\mm_1$ this function vanishes identically but in general we obtain
\begin{equation}\label{eq5_2}
  F'_1=\{H,F_1\}=\frac{1}{\sqrt{x_1 x_2}}\left(\frac{x_2}{x_1}\xi_1+\frac{x_1}{x_2}\xi_2 \right)F.
\end{equation}
We see that the equations
\begin{equation}\label{eq5_3}
  F=0, \qquad F_1=0
\end{equation}
form a pair of invariant relations for the GTFG  system with $\lambda=0$. For the classical Kowalevski case ($\ve_0 =1, \ve_1 =0, p=r=a, G=L_{\rm class}^2$) this system, obviously, turn into the system \eqref{eq3_4} with the integral constants satisfying \eqref{eq3_3}. For the top in a double field the system \eqref{eq5_3} turns into the system found in \cite{Kh32}.

\begin{theorem}\label{theo2}
Let $\mm_2$ denote the closure of the set defined by $\eqref{eq5_3}$. Then $\mm_2$ is a critical subsystem generated by the conditions $\eqref{eq2_12}$ with
\begin{equation}\label{eq5_4}
  \mL=\mL_2= (p^2h-2g+\ve_1 ^2 r^4)^2-r^4 k.
\end{equation}
\end{theorem}

\begin{proof} The set \eqref{eq5_3} is well defined for $x_1 x_2 \ne 0$. Then due to \eqref{eq5_1}, \eqref{eq5_2} its closure is preserved by the phase flow. It is easy to check the following identity
\begin{equation}\label{eq5_5}
  \left[2(p^2 H -2 G+\ve_1 ^2 r^4)-F^2 \right]^2-4 r^4 K \equiv r^4 F_1^2.
\end{equation}
Therefore, $\Phi_{\mL_2}=\mL_2(H,K,G)=0$ on $\mm_2$. Taking the differential of \eqref{eq5_5} we obtain that $\rmd \Phi_{\mL_2}=0$ on $\mm_2$. \end{proof}

Let us also note some features of the subsystem $\mm_2$.

\begin{prop}\label{prop3}
The set $\mm_2$ is a smooth 4-dimensional manifold at least for sufficiently small values of $\ve_1 $.
\end{prop}

\begin{proof} For $\ve_1 =0$ the induced system on $\mm_2$ has an algebraic separation of variables \cite{KhSavUMBeng} and all phase variables are explicitly expressed in terms of two separated variables and two constants of independent partial integrals. These expressions in fact give a straightforward proof of the fact that for $\ve_1 =0$ the invariant set $\mm_2$ is everywhere smooth. Then it is smooth for sufficiently small $\ve_1 $.
\end{proof}

\begin{prop}\label{prop4}
The induced system on $\mm_2$ has the following partial integrals
\begin{eqnarray}
& & M=\frac{1}{2r^2}\left( \frac{x_2}{x_1}\xi_1+\frac{x_1}{x_2}\xi_2 \right)+\ve_1 ^2, \label{eq5_6} \\
& & L=\frac{1}{\sqrt{x_1 x_2}}\left[ (w_1+\rmi \ve_1  z_1) (w_2- \rmi \ve_1  z_2) +(x_1 x_2+z_1 z_2) (M-\ve_1 ^2) \right]+\ve_1 ^2 \sqrt{x_1 x_2}. \label{eq5_7}
\end{eqnarray}
The general first integrals on $\mm_2$ are connected with the partial integrals by the following relations
\begin{equation}\label{eq5_8}
K=r^4(M-\ve_1 ^2)^2, \quad 2G -p^2 H= r^4 M, \quad 2p^2 M^2+2 H M+\ve_0 ^2=L^2.
\end{equation}
The restriction of the symplectic structure to $\mm_2$ degenerates at the zero level of~$L$.
\end{prop}

\begin{proof}
First, notice that on $\mm_2$
\begin{equation*}
  M'=\{H,M\}=-\frac{1}{2 r^2\sqrt{x_1 x_2}} F \,F_1 \equiv 0.
\end{equation*}
It follows from \eqref{eq5_1}, \eqref{eq5_2} and Jacobi's identity that the Poisson bracket $\{F,F_1\}$ is a partial integral on $\mm_2$. According to Lemma~\ref{lem1}, its zeros define the degeneration of the induced \mbox{2-f}orm. At the same time we readily obtain that $\{F,F_1\}=2 r^2 L$.
Finally, the straightforward calculation in virtue of the system \eqref{eq5_3} proves \eqref{eq5_8}.
\end{proof}

Relations \eqref{eq5_8} allow to take either $(H, M)$ or  $(L, M)$ as the pair of independent integrals on $\mm_2$ in order to investigate the corresponding solutions and the phase topology. The form of the integrals \eqref{eq5_6}, \eqref{eq5_7} in the case $\ve_1 =0$ lead to an elliptic separation of variables \cite{KhSavUMBeng}. The last relation in \eqref{eq5_8} shows that the region of existence of motions on the plane of the integral constants is partly bounded by the image of the set on which the induced symplectic structure degenerates. It causes new topological effects of non-orientable bifurcations. Considering an appropriate path crossing the set $\{L=0\}$ it can be shown that $\mm_2$ as a whole is non-orientable.

\section{The extension of the 4th Appelrot class}\label{sec5}
The 4th Appelrot class of motions has the widest range of generalizations. One of its main features is that the projection of the angular velocity onto the axis containing the mass center is constant. Taking this as a starting condition, P.V.\,Kharlamov \cite{PVLect} found the partial integrable case for an axially symmetric gyrostat in the gravity field, which covers also the known Bobylev\,--\,Steklov case. The attempts to generalize relations \eqref{eq3_7} to the case of a double field, as was shown by H.M.\,Yehia \cite{Yeh2}, lead only to the trivial class of pendulum type motions. While investigating critical motions of the top in a double field \cite{KhRCD}, it became clear that for a wide class of motions the first condition \eqref{eq3_6} holds, i.e., despite the fact that the scalar product $\mathbf{M}{\bdot}\bda$ is no longer constant, there exists a 4-dimensional critical submanifold on which the ratio $\mathbf{M}{\bdot}\bda/M_1$ remains constant. Then it leads to the condition that $\mathbf{M}{\bdot}\bdb/M_2$ is also constant and these two ratios stay equal to each other along the trajectories. This condition naturally appeared in \cite{KhRCD} after excluding previously found critical subsystems \cite{Bogo,Kh32} from the equations of critical points of the integral map. Later this condition was used in \cite{KhHMJ} to construct the extension of the 4th Appelrot class for the gyrostat in a double field. Therefore, for the problem of the generalized two fields gyrostat we start from the same assumption, which can be written in the form $U_1 = 0$, where
\begin{equation*}
\begin{array}{l}
 \displaystyle  U_1 = \frac{\mathbf{M}{\bdot}\bda+\rmi \mathbf{M}{\bdot}\bdb}{M_1+\rmi M_2}- \frac{\mathbf{M}{\bdot}\bda-\rmi \mathbf{M}{\bdot}\bdb}{M_1-\rmi M_2}\\[3mm]
 \displaystyle  \phantom{U_1} = \frac{w_2^2 x_1 - w_1^2 x_2 -  w_1 w_2(y_1 - y_2) + (z_1 w_2 - w_1 z_2) (w_3 + \lambda)}{2 w_1 w_2}.
\end{array}
\end{equation*}
Let us define the function $U_2 = 4 w_1 w_2 U'_1$ and calculate its derivative by virtue of equations \eqref{eq2_1}. We come to the following system
\begin{equation}\label{eq6_1}
\displaystyle  U'_1 = \frac{1}{4 w_1 w_2} U_2, \qquad U'_2=\frac{\varkappa}{w_1 w_2}  U_1,
\end{equation}
where $\varkappa$ is a polynomial in variables \eqref{eq2_14}. Both $U_2$ and $\varkappa$ can be easily written out with the help of computer algebra. Here we omit their full expressions but use them for various calculations.

\begin{theorem}\label{theo3}
The closure $\mm_3$ of the set defined by the equations
\begin{equation}\label{eq6_2}
  U_1=0, \qquad U_2=0
\end{equation}
is an invariant set of the~GTFG  system. This set consists of the critical points of the integral map and is almost everywhere a 4-dimensional manifold. The induced system on $\mm_3$ has  partial integrals
\begin{equation}\label{eq6_3}
\begin{array}{l}
   \displaystyle   S= -\frac{1}{2}\left[\frac{\mathbf{M}{\bdot}\bda+\rmi \mathbf{M}{\bdot}\bdb}{M_1+\rmi M_2}+ \frac{\mathbf{M}{\bdot}\bda-\rmi \mathbf{M}{\bdot}\bdb}{M_1-\rmi M_2}\right]\\[3mm]
  \displaystyle  \quad = -\frac{1}{4 w_1 w_2 } \left[ w_1^2 x_2 +w_2^2 x_1+ w_1 w_2 (y_1+y_2)+(z_1 w_2 +z_2 w_1 ) (w_3+\lambda)\right], \\[3mm]
   \displaystyle   T = x_1 x_2 +z_1 z_2 - \frac{2 S}{\ve_0 +2\ve_1 ^2 S}\bigl[(w_1+\rmi\,\ve_1  z_1)(w_2-\rmi\,\ve_1  z_2)+\ve_1 ^2 x_1 x_2\bigr].
\end{array}
\end{equation}
On $\mm_3$, the general integrals $H,G,K$ are expressed in terms of $S,T$ by the following identities
\begin{equation}\label{eq6_4}
\begin{array}{l}
\displaystyle  H= \ve_0 S + \frac{\ve_0 }{2 S}(p^2-T)- \ve_1 ^2 T +\frac{\lambda^2}{2}, \\[3mm]
\displaystyle  K = \frac{(\ve_0 +2\ve_1 ^2 S)^2\bigl[ r^4+T(T+4S^2-2p^2)\bigr]}{4 S^2}, \\[3mm]
\displaystyle  G= \frac{\ve_0 (p^4-r^4+2p^2 S^2)}{4 S}- \frac{1}{2}S (\ve_0 +2\ve_1 ^2 S)T.
\end{array}
\end{equation}
\end{theorem}
\begin{proof}
It follows from \eqref{eq6_1} that equations \eqref{eq6_2} define an invariant set in the domain $w_1 w_2 \ne 0$. Then its closure $\mm_3$ is also preserved by the phase flow.
Obviously, $\mm_3$ is a smooth 4-dimensional manifold except for the points where $\rk (U_1,U_2) < 2$. These points are defined by equations \eqref{eq2_16} with
\begin{equation*}
f= \rmi \mu_1 U_1+\mu_2 U_2 \qquad (\mu_1,\mu_2 \in \bbR, \; \mu_1^2+\mu_2^2 \ne 0).
\end{equation*}
Since we deal with polynomials and equations \eqref{eq2_16} are not identities at least for $\ve_1 =0$, they define a subset of positive codimension in $\mm_3$.
In fact, equations \eqref{eq2_16} in this case give the following condition $w_1 w_2 Q=0$,
where
\begin{equation*}
Q = \ve_1  \left[ (z_2 w_1-z_1 w_2)(w_3+\lambda)+w_1^2 x_2 - w_2^2 x_1 \right] +\rmi \, w_1 w_2 (w_3-\lambda).
\end{equation*}
Therefore, $\mm_3$ is smooth and 4-dimensional almost everywhere.

By virtue of \eqref{eq2_1} we have
\begin{equation*}
\begin{array}{l}
\displaystyle S'=-\frac{1}{4 w_1 w_2}\bigl\{w_1 w_2 (w_3-\lambda)+\ve_0 (z_1 w_2+z_2 w_1)\\[3mm]
\displaystyle \quad -\rmi\, \ve_1 [w_1^2 x_2 - w_2^2 x_1- w_1 w_2(y_1-y_2)] \bigr\}U_1,\\[3mm]
\displaystyle T'=\frac{z_1 w_2+z_2 w_1}{4w_1 w_2(\ve_0 +2\ve_1 ^2S)}\left\{ \rmi\, \ve_1  U_2 +2 \bigl[ \ve_1 ^2(w_1^2 x_2+w_2^2 x_1) \right.\\[3mm]
\displaystyle \quad +\left. \rmi\, \ve_0 \ve_1  (z_2 w_1 - z_1 w_2)-\ve_0 w_1 w_2 \bigr] U_1 \right \},
\end{array}
\end{equation*}
so $S$ and $T$ are partial integrals on $\mm_3$.

To check various equalities on $\mm_3$, let us solve equations \eqref{eq6_2} with respect to $y_1,y_2$. Supposing that $w_1 w_2 Q\ne 0$, we have
\begin{equation}\label{eq6_5}
\begin{array}{l}
   \displaystyle y_1 = \frac{1}{w_2 Q} \Bigl\{ \ve_1  w_1 [w_1 x_2 + z_2 (w_3 + \lambda)]^2-  w_2^2 x_1 [\ve_1  w_1 x_2 + \ve_1  z_2 (w_3 + \lambda)-\rmi \ve_0 z_2] \\[3mm]
  \qquad - \rmi \, w_2 [w_1 x_2 + z_2 (w_3 + \lambda)] [w_1 (w_3 - \lambda)+ \ve_0 z_1 + \rmi \ve_1  z_1 (w_3 + \lambda)]  \Bigr\},\\[3mm]
   \displaystyle y_2 = \frac{1}{w_1 Q} \Bigr\{ \ve_1  w_2 [w_2 x_1 + z_1 (w_3 + \lambda)]^2-  w_1^2 x_2 [\ve_1  w_2 x_1 + \ve_1  z_1 (w_3 + \lambda)+\rmi \ve_0 z_1]\\[3mm]
  \qquad - \rmi \, w_1 [w_2 x_1 + z_1 (w_3 + \lambda)] [w_2 (w_3 - \lambda)+ \ve_0 z_2 - \rmi \ve_1  z_2 (w_3 + \lambda)]\Bigr\}  .
\end{array}
\end{equation}
These values along with \eqref{eq6_3} turn \eqref{eq6_4} into identities. This also proves that the general integrals $H,K,G$ are dependent on $\mm_3$, so $\mm_3$ with the induced flow on it is a critical subsystem.
\end{proof}

Let us emphasize an interesting fact. Replacing in \eqref{eq6_4} the integrals $H,K,G,S,T$ with their constants, we come to the parametric equations of the bifurcation surface in $\bbR^3(h,k,g)$ corresponding to the critical subsystem $\mm_3$. The parameters are the constants of $S$ and $T$. Note that after the energy shift $h_*= h - \frac{\lambda^2 }{2}$ the obtained equations become independent of the gyrostatic momentum~$\lambda$. Similar result for the gyrostat in a double field was mentioned in \cite{KhHMJ}. To compare with the classical Kowalevski case we put $\ve_0 =1,\ve_1 =0, p=r=a^2$ and eliminate $T$ in \eqref{eq6_4} with the help of the first equation. Since in this case $G=L_{\rm class}^2$ we come to the surface \eqref{eq3_5}.

\begin{prop}\label{prop5}
The Poisson bracket
\begin{equation}\label{eq6_6}
  \{U_1,U_2\}=\frac{2}{S}\left[ 12 \ve_0 S^4 - 8 \bigl(h-\frac{\lambda^2}{2}\bigr) S^3 + \ve_0 (p^4-r^4)\right]
\end{equation}
is a partial integral of the critical subsystem $\mm_3$. Its zero level defines the set of degeneration of the 2-form on $\mm_3$ induced by the symplectic structure of $\mP^6$. The image of this level under the integral map is a cuspidal edge of the corresponding bifurcation surface in $\bbR^3(h,k,g)$.
\end{prop}

\begin{proof}
It follows from equations \eqref{eq6_1} and Jacobi's identity that the Poisson bracket of $U_1$ and $U_2$ is a partial integral on the manifold \eqref{eq6_2}. The expression \eqref{eq6_6} is checked by substituting the values of $y_1,y_2$ from \eqref{eq6_5}. The degenerations of the induced 2-form are defined by the equation $\{U_1,U_2\}=0$ according to Lemma~\ref{lem1}. The last statement immediately follows from \eqref{eq6_4}.
\end{proof}

Obviously, the obtained critical subsystem $\mm_3$ can be also generated by the equation of the type $\eqref{eq2_10}$. To find the corresponding function $\mL_3(h,k,g)$ replace in $\eqref{eq6_4}$ the general integrals $H,K,G$ with the corresponding constants, write the equations in polynomial form with respect to $h,k,g,S,T$ and calculate the resultants to eliminate $S,T$. The final expression is too complicated and hardly can be used for any practical purpose. In \cite{RyabTMF} for the case $\lambda=0$ the surface $\mL_3=0$ is represented as a discriminant surface of some polynomial with coefficients depending on $h,k,g$.

\section{Applications}\label{sec6}
Let us demonstrate an application of the above results to calculate the types of critical points of the integral map. Consider a critical subsystem $\mm \subset \mP^6 \subset \mathbf{v}^*$ and let $\Phi$ be the first integral such that $\Phi$ is regular in some neighbourhood of $\mm$ except for the points of $\mm$, i.e.,
$$
\mm=\{x \in \mP^6: \Phi(x)=0, \,\rmd \Phi(x) =0\}.
$$
As it was mentioned above, such integrals usually appear from the equations of bifurcation surfaces like in \eqref{eq2_10} -- \eqref{eq2_12}, but also can be obtained from the equations of $\mm$. Each point $x\in \mm$ is a singular point of the Hamiltonian vector field $\sgrad \Phi$ on $\mathbf{v}^*$ obtained from \eqref{eq2_1} by replacing $H$ with~$\Phi$. Then the linearization of this field at the point $x$ restricted to $T_x \mP^6$ is a well-defined symplectic operator denoted by $A_\Phi$. If calculated in the initial variables on $\mathbf{v}^*$, this operator gets three more zero eigenvalues corresponding to the Casimir functions. If $\rk J(x)=2$, then the characteristic polynomial of $A_\Phi$ has the form
\begin{equation}\label{eq7_1}
  \chi_\Phi(\mu) = - \mu^7 (\mu^2-C_\Phi),
\end{equation}
in which $C_\Phi$ is constant on the whole integral manifold containing $x$. Since $\mm$ is the closure of the set of points of rank 2, this form is valid through all of $\mm$.
Then the point $x$ has the type ``center'' if $C_\Phi<0$, the type ``saddle'' if $C_\Phi>0$, and $x$ is said to be degenerate if $C_\Phi=0$. This terminology reflects the type of $x$ as a singular point of the projection of the phase flow onto 2-dimensional section transversal to $\mm$ in $\mP^6$ at the point $x$. For the induced system on $\mm$ points of rank~2 are regular. Therefore, the above defined type can be called \textit{the outer type} of a critical point $x$ with respect to $\mm$. If we can calculate all outer types for all critical subsystems, then, for example, a point of rank~1, which belongs to an intersection of two critical subsystems, will have two outer types and they define the complete type with respect to the whole system on $\mP^6$. Similarly, to obtain the type of a point of rank 0, we must know three outer types of this point with respect to three intersecting critical subsystems. Note that singularities of the type ``focus'' in a system with three degrees of freedom should be of rank not more than 1 and, if exist, are organized in 2-dimensional critical subsystems. In the GTFG system ``focus'' type singularities form a restricted set of the pendulum type motions which has no analogue in the classical case but exist for the gyrostat in a double field \cite{KhHMJ,KhRyabSmir41}.

Suppose that $\mm$ is defined by two invariant relations \eqref{eq4_8}. For any point $x$ we can  calculate $\sgrad \Phi(x)$ and the characteristic polynomial $\chi_\Phi$, but only equations \eqref{eq4_8} will guarantee that it has the form \eqref{eq7_1} with $C_\Phi$ depending only on the integral constants. For points of rank 2 we have to express a pair of phase variables from \eqref{eq4_8}, substitute the obtained expressions into $\chi_\Phi$ and find the value of $C_\Phi$. In fact, there is no need to calculate $\chi_\Phi$ completely; $C_\Phi$ is the coefficient of $\mu^7$ in this case equal to $\Tr (A_\Phi^2)/2$.

Consider the critical subsystem $\mm_1$. According to \eqref{eq4_2} we take $\Phi_1=K$. Due to the specific structure \eqref{eq4_3} of this function we see that small positive levels of $K$ on the transversal section to $\mm_1$ look like circles. Then the outer type of the points of $\mm_1$ cannot be ``saddle''. Still these points can be degenerate. To simplify $C_1=C_{\Phi_1}=\Tr(A_{\Phi_1}^2)/2$ let us introduce new variables
\begin{equation*}
  W_1 = w_1 +\rmi\, \ve_1  z_1, \qquad W_2 = w_2 -\rmi\, \ve_1  z_2.
\end{equation*}
Relations \eqref{eq4_7} and the partial first integrals \eqref{eq4_10}, \eqref{eq4_11} take the form
\begin{equation*}
\begin{array}{c}
  W_1^2=-x_1 [\ve_0 -\rmi\,\ve_1 (w_3-\rmi\,\ve_1  y_2)], \quad W_2^2=-x_2 [\ve_0 +\rmi\,\ve_1 (w_3+\rmi\,\ve_1  y_1)], \\
  \displaystyle N = \frac{W_1 W_2}{\sqrt{x_1 x_2}} +\ve_1  ^2\sqrt{x_1 x_2}, \quad    \displaystyle F = \sqrt{x_1 x_2} [w_3 +\rmi\, \ve_1  (y_1-y_2)] - \frac{1}{\sqrt{x_1 x_2}} (x_2 z_1 W_1+x_1 z_2 W_2) 
\end{array}
\end{equation*}
and allow to eliminate the higher powers of $W_1,W_2$, the product $W_1 W_2$ and the variable $w_3$. Finally we obtain
\begin{equation*}
  C_1 = C_1(N,F) = -4 N^2 F^2.
\end{equation*}

\begin{theorem}\label{theo4}
Points of the critical subsystem $\mm_1$ have the outer type ``center'' except for the points defined by the equation $N\,F=0$, which are degenerate with respect to the whole system on~$\mP^6$.
\end{theorem}

Comparing with Proposition~\ref{prop2} and equation \eqref{eq4_12} we see that degenerate points of $\mm_1$ (as far as the outer type is considered) form exactly the set of degeneration of the 2-form induced by the symplectic structure.


Now let us analyze the points of the second critical subsystem $\mm_2$.
Similar to the previous case, some general statements on the outer type can be established immediately from the first integrals.
To simply the formulas, we introduce the partial integral shifted from \eqref{eq5_6} $M_0=M-\ve_1^2$. Let $m_0$ denote its constant on $\mm_2$. We have the following identities
\begin{equation}\label{eq7_2}
\begin{array}{l}
\displaystyle  2 G -p^2 H -\ve_1 ^2 r^4 +\frac{1}{2} F^2 = r^4 M_0, \qquad \displaystyle  r^4 M_0^2 - K =\frac{1}{4} F_1^2.
\end{array}
\end{equation}
Let
\begin{equation*}
Q_1=\sqrt{\frac{x_2}{x_1}\xi_1} \quad {\rm and} \quad Q_2=\sqrt{\frac{x_1}{x_2}\xi_2}
\end{equation*}
be complex conjugate. Introduce the following real functions
\begin{equation*}
  \Psi_+ = \frac{Q_1+Q_2}{2}, \qquad \Psi_- = \frac{Q_1-Q_2}{2\rmi}.
\end{equation*}
Then $F_1=0$ yields either $\Psi_+=0$ or $\Psi_-=0$. Replacing the general integrals $H,K,G$ with its constants and choosing the arithmetic value $\sqrt{k} \geqslant 0$, we obtain from \eqref{eq7_2} the following identities valid in the whole $\mP^6$
\begin{eqnarray}
\displaystyle  (p^2 h - 2 g  +\ve_1 ^2 r^4) + r^2 \sqrt{k} = \frac{1}{2}F^2 + 2 r^2 \Psi_-^2, \label{eq7_3}\\
\displaystyle  (p^2 h - 2 g  +\ve_1 ^2 r^4) - r^2 \sqrt{k} = \frac{1}{2}F^2 - 2 r^2 \Psi_+^2. \label{eq7_4}
\end{eqnarray}
Thus, the subsystem $\mm_2$ generated by the equation $(p^2 h - 2 g  +\ve_1 ^2 r^4)^2 - r^4 {k} =0$ naturally splits into two parts. On the first part corresponding to equation \eqref{eq7_3} similar to the 2nd Appelrot class we have $p^2 h - 2 g  +\ve_1 ^2 r^4<0$. Taking $F,\Psi_-$ for the local coordinates in the transversal section to $\mm_2$ we obviously obtain the type ``center'' just looking at the right-hand part of \eqref{eq7_3}. The second part of $\mm_2$ is given by equation \eqref{eq7_4} and satisfies the condition $p^2 h - 2 g  +\ve_1 ^2 r^4>0$ similar to the 3rd Appelrot class. Then from the right-hand part of equation \eqref{eq7_4} taking $F,\Psi_+$ for the local coordinates in the transversal section to $\mm_2$ we obviously obtain the type ``saddle''.
Since on $\mm_2$ according to \eqref{eq7_2} the sign of $p^2 h - 2 g  +\ve_1 ^2 r^4$ is opposite to the sign of the integral constant $m_0$, the type is ``center'' if $m_0>0$ and ``saddle'' if $m_0<0$. Nevertheless, we have to check this guess by straightforward calculation, because we did not discuss the possibilities of the involved functions to be dependent in the sense of differentials.

For this case we can take the function $\Phi_{\mL_2}$ generated by \eqref{eq5_4} according to the rule \eqref{eq2_11}. The calculation of the terms of the characteristic polynomial under such choice is too complicated. Let us use the identities \eqref{eq7_2}. Then we have $\Phi_{\mL_2}=\frac{1}{4}(F^4+r^4 F_1^2-4r^4 F^2 M_0)$. Since on $\mm_2$ according to \eqref{eq5_3} $A_{F^4}\equiv 0$, we can take
\begin{equation*}
\Phi_2=\frac{F_1^2}{4} -F^2 M_0.
\end{equation*}
To simplify the expression of $C_2=\Tr(A_{\Phi_2}^2)/2$ it is convenient to solve equations \eqref{eq5_3} with respect to $y_1,y_2$:
\begin{equation*}
\begin{array}{l}
\displaystyle  y_1=\frac{1}{\ve_1 ^2 x_1 x_2(x_1-x_2)}\left\{x_1 W_1^2-x_2 W_2^2+\rmi\,\ve_1  x_2\bigl( x_2 z_1 W_1+x_1 z_2 W_2\bigr)  \right.\\[3mm]
\displaystyle  \left. \quad +x_1 x_2 \bigl[\ve_1  (x_1-x_2)+\rmi \ve_1  x_1 w_3\bigr]\right\}, \\[3mm]
\displaystyle  y_2=\frac{1}{\ve_1 ^2 x_1 x_2(x_1-x_2)}\left\{x_1 W_1^2-x_2 W_2^2+\rmi\,\ve_1  x_1\bigl( x_2 z_1 W_1+x_1 z_2 W_2\bigr)  \right.\\[3mm]
\displaystyle  \left. \quad +x_1 x_2 \bigl[\ve_1  (x_1-x_2)+\rmi \ve_1  x_2 w_3\bigr]\right\}.
\end{array}
\end{equation*}
After substitution of these values into $C_2$, the latter gets the following representation
\begin{equation*}
  C_2 = - 4 r^{4} M_0 L^2.
\end{equation*}
Thus, we obtain the complete result.
\begin{theorem}\label{theo5}
The outer type of the points of the critical subsystem $\mm_2$ is defined in terms of the partial integrals $M$ and $L$ as follows.
If $M=\ve_1 ^2$ or $L=0$ the points are degenerate; if $L \ne 0$ then in the part $M > \ve_1 ^2$ all points have the outer type ``center'' and in the part $M < \ve_1 ^2$ all points have the outer type ``saddle''.
\end{theorem}

Again we see that the statement obtained \textit{a priory} from the first integrals should be corrected by adding to the set of degenerate points the whole set $L=0$, on which the form induced by the symplectic structure is degenerate.


Consider the critical subsystem $\mm_3$. As we mentioned above, equation~\eqref{eq2_10} for the corresponding surface is too complicated. Nevertheless, having the identities \eqref{eq6_4} we easily get on $\mm_3$ the following dependency of the differentials with coefficients expressed in terms of the partial integrals
\begin{equation*}
\displaystyle  2 \rmd G -(p^2-T) \rmd H +\frac{S}{\ve_0 +2\ve_1 ^2 S} \rmd K=0.
\end{equation*}
Therefore, we can take
\begin{equation*}
  \Phi_3 = \displaystyle  2 G -(p^2-T) H +\frac{S}{\ve_0 +2\ve_1 ^2 S} K + {\rm const},
\end{equation*}
calculate the needed term of the characteristic polynomial considering $S,T$ as constants, and \textit{after} that make the substitution of $S,T$ from \eqref{eq6_3} and of $y_1,y_2$ as the solutions of the invariant relations \eqref{eq6_2} from \eqref{eq6_5}. Then for the eigenvalues defining the outer type we have $\mu^2 =C_3$, where
\begin{equation*}
\begin{array}{l}
  \displaystyle  C_3(S,T)= \frac{\ve_0 \bigl[ (p^2-2S^2)^2-r^4\bigr]+4 ({\ve_0 +2\ve_1 ^2 S}) S^2 T}{2({\ve_0 +2\ve_1 ^2 S})^2 S} \left[(\ve_0 +2\ve_1 ^2 S)^2 T -2 \ve_0 \lambda^2 S \right].
\end{array}
\end{equation*}
Since according to the first equation in \eqref{eq6_4} the functions $T$ and $H$ on $\mm_3$ are connected by linear dependency, we can also express the value $\mu^2$ in terms of $H,S$
\begin{equation}\label{eq7_5}
\begin{array}{l}
  \displaystyle  C_3(H,S) = \frac{1}{2({\ve_0 +2\ve_1 ^2 S})^2 S}\left[ 12 \ve_0 S^4 - 8 \bigl(h-\frac{\lambda^2}{2}\bigr) S^3 + \ve_0 (p^4-r^4)\right] {\times} \\[3mm]
  \displaystyle  \qquad {\times} \left\{ (\ve_0 +2\ve_1 ^2 S) \bigl[\ve_0 (p^2+2 S^2)-2 (h-\frac{\lambda^2}{2})S \bigl]  -2 \ve_0 \lambda^2 S  \right\}.
\end{array}
\end{equation}
This representation is useful when iso-energetic bifurcation diagrams and corresponding bifurcations are investigated. Thus, we obtain the following result.

\begin{theorem}\label{theo6}
The outer type of the points of the critical subsystem $\mm_3$ is defined in terms of the first integrals by the sign of $C_3$. The points with $C_3<0$ have the outer type ``center'', and the points with $C_3>0$ have the outer type ``saddle''. If $C_3=0$, then the points are degenerate critical points with respect to the system on $\mP^6$.
\end{theorem}

Comparing \eqref{eq6_6} and \eqref{eq7_5} we see that, similar to both previous subsystems, all points at which the induced symplectic structure degenerates are degenerate for the complete system due to their outer type with respect to $\mm_3$. In fact one can prove the following general statement.

\begin{prop}\label{prop55}
Suppose a submanifold $\mm$ of a symplectic manifold $\mathcal{N}$ is defined by the system of independent equations $f_1=0,\,f_2 =0$. Let $\mm=\{x\in \mathcal{N}: {\Phi(x)=0}, \,{\rmd \Phi(x) =0}\}$ for some function $\Phi$ on $\mathcal{N}$. Let $\mu^2$ be the square of the eigenvalues of the symplectic operator $A_\Phi$ restricted to a 2-dimensional section transversal to $\mm$. Then $\{f_1,f_2\}=0$ yields $\mu^2=0$.
\end{prop}
The strict proof of this statement is not in the frame of this article. In the above examples we saw that $\{f_1,f_2\}$ is a multiplier of $C_\Phi$.

\section{Conclusion}\label{sec7}
In this paper we present the explicit equations (invariant relations) of the invariant manifolds of the generalized two field gyrostat system in three cases which extend the famous Appelrot classes of motions of the classical Kowalevski top. Earlier such extensions were obtained for the line of the Kowalevski type tops and gyrostats, including the Kowalevski\,--\,Yehia gyrostat, the Kowalevski\,--\,Bogoyavlensky top in a double field and the Kowalevski\,--\,Reyman\,--\,Semenov-Tian-Shansky gyrostat in a double field.
We point out the number of partial integrals convenient to take as independent integrals while investigating the phase topology of the found critical subsystems and express via these integrals the general integrals of the initial system.

We also reveal some important characteristics of the obtained invariant manifolds, such as smoothness and degeneration of the induced symplectic structure. For all three subsystems we give the explicit equation for the points of degeneration of this \mbox{2-form} by expressing the Poisson bracket of the invariant relations in terms of the first integrals.

It is shown how to use the invariant relations to calculate the type of critical points of the integral map. Such calculation gives main analytical foundation to establish the phase topology of the whole system with three degrees of freedom and provides a useful appendix to the work \cite{RyabTMF}, in which the highly complicated topological analysis of the considered Sokolov\,--\,Tsiganov case was started.

Two other invariant manifolds are known for the GTFG system. One of them consists of the pendulum motions in which the constant direction of the kinetic moment is orthogonal to the plane of the forces intensities \cite{Yeh2}. These mentioned above motions fill a \mbox{2-dim}ensional invariant manifold and can have the type ``focus''. One more 4-dimensional critical manifold was found in \cite{RyabTMF} and does not have any analogues. It exists only for large enough values of the parameter $\ve_1$. Nevertheless, until now it is not proved even for the case $\lambda=0$ that the known subsystems cover completely the critical set of the integral map. The algebraic curve of the Lax representation given in \cite{SokTsiE} always have singular points and does not provide any proofs also. It seems a very interesting problem to find an exhaustive classification of the critical subsystems along with complete classification of the critical points with respect to all parameters of the system.

{\bf Acknowledgements.} This work was partially supported by RFBR and Volgograd Region Authority, research projects No.~13-01-97025, 14-01-00119.


\def\vol{V.~}
\def\jr{ }
\def\pg{P.~}
\def\au{\textit}
\def\nn{No~}

\end{document}